\pdfoutput=1
\newif\ifieee
\ieeetrue
\ifieee
    \documentclass[letterpaper, 10 pt, conference]{ieeeconf} 
    \IEEEoverridecommandlockouts                              
\else
    \documentclass{article}
    \usepackage[margin=1.25in]{geometry}
    \usepackage{amsthm} 
\fi

\usepackage{graphicx} 
\usepackage{xcolor} 
\usepackage{caption}
\usepackage{subcaption}
\usepackage{tikz} 

\usepackage{amsmath, amsfonts, amssymb, amsthm}
\usepackage{cite}
\usepackage{algorithm}
\usepackage{algpseudocode}

\usepackage{mathtools}  
\mathtoolsset{showonlyrefs}  

\usepackage{hyperref}

\usepackage{fancyhdr}
\fancypagestyle{arxiv}{%
  \fancyhf{} 
  
  \setlength{\headheight}{20pt}   
  \fancyhead[C]{Joint submission to IEEE Control Systems Letters and the 2024 American Control Conference}
}%

\newtheorem{definition}{Definition}
\newtheorem{theorem}{Theorem}

\newtheorem{corollary}{Corollary}
\newtheorem{remark}{Remark}
\newtheorem{problem}{Problem}

\newtheorem{example}{Example}

\DeclareMathOperator*{\argmax}{arg\,max}
\DeclareMathOperator*{\argmin}{arg\,min}
\algnewcommand\algorithmicforeach{\textbf{for each}}
\algdef{S}[FOR]{ForEach}[1]{\algorithmicforeach\ #1\ \algorithmicdo}

\newcommand{\imc}{\mathcal{I}}

\newcommand{\reals}{\mathbb{R}}
\newcommand{\plow}{\check p}
\newcommand{\phigh}{\hat p}
\newcommand{\Plow}{\check P}
\newcommand{\Phigh}{\hat P}

\newcommand{\IStates}{Q}
\newcommand{\Istate}{q}

\newcommand{\Indy}{\mathbf{1}}
\newcommand{\x}{\mathbf{x}}
\newcommand{\w}{\mathbf{w}}
\newcommand{\adv}{\gamma}
\newcommand{\Adv}{\Gamma}
\newcommand{\dist}{p(\w)}

\renewcommand{\xi}{\x^i}
\newcommand{\wi}{\w^i}
\newcommand{\Istatei}{\Istate^{'i}}
\newcommand{\Cupper}{\overline C^*}
\newcommand{\Clower}{\underline C^*}
\newcommand{\Post}{Post}

\title{\LARGE \bf
Formal Abstraction of General Stochastic Systems via Noise Partitioning
}

\author{John Skovbekk$^{1}$, Luca Laurenti$^{2}$,  Eric Frew$^{1}$, and Morteza Lahijanian$^{1}$
\thanks{This work was supported in part by NSF grant 2039062.}
\thanks{$^{1}$John Skovbekk, Eric Frew, and Morteza Lahijanian are with the Smead Department of Aerospace      Engineering Sciences, CU Boulder, Boulder, CO, USA
        {\tt\small \{firstname.lastname\}@colorado.edu}}%
\thanks{$^{2}$Luca Laurenti is with the Center for Systems and Control at TU Delft,
        Delft, Netherlands
        {\tt\small L.Laurenti@tudelft.nl}}%
}

\begin{document}

\maketitle
\thispagestyle{arxiv} 
\pagestyle{arxiv}

\begin{abstract}
Verifying the performance of safety-critical, stochastic systems with complex noise distributions is difficult.
We introduce a general procedure for the finite abstraction of nonlinear stochastic systems with non-standard (e.g., non-affine, non-symmetric, non-unimodal) noise distributions for verification purposes.
The method uses a finite partitioning of the noise domain to construct an interval Markov chain (IMC) abstraction of the system via transition probability intervals. 
Noise partitioning allows for a general class of distributions and structures, including multiplicative and mixture models, and admits both known and data-driven systems.
The partitions required for optimal transition bounds are specified for systems that are monotonic with respect to the noise, and explicit partitions are provided for affine and multiplicative structures.
By the soundness of the abstraction procedure, verification on the IMC provides guarantees on the stochastic system against a temporal logic specification.
In addition, we present a novel refinement-free algorithm that improves the verification results.
Case studies on linear and nonlinear systems with non-Gaussian noise, including a data-driven example, demonstrate the generality and effectiveness of the method without introducing excessive conservatism.
    
\end{abstract}

\section{INTRODUCTION}
The deployment of autonomous systems for safety-critical applications, such as medical robotics and self-driving vehicles, requires diligent verification of their behavior.  
Such systems are inherently stochastic due to uncertainty in physical components (e.g., noise in sensors and actuators) or black-box software components.
Formal methods provides rigorous techniques for verifying stochastic systems subject to temporal logic specifications~\cite{BaierBook2008, lavaei2022automated}.
In particular, powerful model checking algorithms exist for finite-state Markov processes that can scale to large systems~\cite{BaierBook2008}.
However, to apply them to continuous-space systems, finite abstractions with correctness guarantees are required \cite{lavaei2022automated,alur2000discrete}, which is difficult in both accuracy and scalability. 
For this reason, most existing work focuses on specific classes of stochastic systems often with strong assumptions on the dynamics or noise models \cite{lahijanian2015formal, jackson2021strategy, dutreix2022abstraction, badings2023robust}, which we aim to relax in this work. 



Uncertain Markov models, namely interval Markov chains (IMCs \cite{givan2000bounded}) have proven to be effective abstraction models for stochastic systems~\cite{lahijanian2015formal,cauchi2019efficiency,jackson2021strategy,jiang2022safe,dutreix2022abstraction}. 
Beyond capturing stochasticity, they also provide a means to incorporate other sources of uncertainty (e.g., discretization error), thereby facilitating correctness. 
Yet, the difficulty remains for general stochastic models due to the need to correctly compute stochastic transition kernels.
Existing techniques rely on standard (unimodal, symmetric and zero-mean) or affine noise distributions~\cite{jackson2021formal, jiang2022safe, dutreix2022abstraction, vanHuijgevoort2023temporal}, linear systems~\cite{cauchi2019efficiency, badings2023robust}.
Additionally, stochastic systems may possess multiple sources of uncertainty, such as data-driven settings~\cite{jackson2021strategy, hashimoto2022learning, jiang2022safe, adams2022formal,badings2023robust}.
Thus, IMC abstraction approaches for nonlinear systems that admit a wider class of distributions and structures are necessary to lift these limitations.

Another difficulty facing abstraction is the state-explosion dilemma in higher dimensions.  
Common approaches to this problem are focused on parallelizing computation~\cite{soudjani2015faust} and adaptive refinement~\cite{lahijanian2015formal, esmaeil2013adaptive}. 
Despite these efforts, the state-explosion problem remains, and new ideas are needed for further mitigation.
Specifically, using the continuous system in tandem with the abstraction to improve the verification without refinement is largely unexplored.

\textbf{Contributions:} We present an abstraction method for nonlinear stochastic systems with non-affine, non-standard noise that admits known and data-driven systems.
Our method generalizes an approach for systems learned from data with affine, sub-Gaussian noise~\cite{jackson2021strategy}.
It is based on partitioning the noise domain to bound the transition kernel of the IMC, side-stepping the need to evaluate it.
We show optimality criteria for the noise partitions for systems with noise monotonicity, and provide explicit partitions for affine and multiplicative structures.
To help address the state-explosion problem, we also propose a refinement-free method to improve the verification results of an abstraction by using the continuous process.  
Finally, we demonstrate the efficacy of the method by verifying linear, nonlinear and data-driven systems without introducing excessive conservatism.

In summary, our contributions are (1) a procedure for constructing abstractions via noise partitioning (Theorem \ref{thm:transition_bounds}); (2) optimal noise partition sizes and values for a general class of distributions (Theorem \ref{thm:partitions}); (3) a procedure to improve the verification of the abstraction without refinement (Algorithm \ref{alg:clustering}), and (4) evaluations and applications to nonlinear systems with non-standard and multiplicative noise (Section \ref{sec:examples}).

\section{PROBLEM FORMULATION}\label{sec:problem}
We first introduce the stochastic process and its finite abstraction, and then formulate two main problems.

\subsection{Stochastic Process Model}

Consider the following discrete-time stochastic process 
\begin{equation}\label{eq:system}
    \x(k+1) = f(\x(k), \w(k)),
\end{equation}
where $\x\in\reals^n$, $\w \in W \subseteq \reals^{n_\w}$ is i.i.d. process noise sampled from distribution $\dist$ with possibly bounded support, and $f: \reals^n \times W \to \reals^n$ is a possibly nonlinear function.
Distribution $\dist$ is allowed to be non-standard, i.e., non-uniform and non-symmetric.
Let $X\subset\reals^n$ be a Borel measureable set.
The one-step transition kernel, which defines the probability of $\x(k+1) \in X$ given $\x(k)=x_k$ is
\begin{equation}\label{eq:transition_kernel}
    T(X\mid x_k) = \int_{X} f(x_k, \w(k)) p(\w) d\w .
\end{equation}
The transition kernel $T$ is the basis for probability measures of paths of System~\eqref{eq:system}~\cite{Klenke2008}, i.e.,
given an initial condition $\x(0)=x_0$,
$\Pr\big(\x(0) \in X\mid x_0\big) = \Indy(x_0\in X)$ and
$\Pr\big(\x(k+1) \in X\mid x_{k}\big) = T(X\mid x_k)$ 
where $\Indy(\cdot)$ is the indicator function that returns 1 if the argument is true and 0 otherwise.
\subsection{Interval Markov Chains}
A finite abstraction of System~\eqref{eq:system} is often an interval Markov chain~\cite{lavaei2022automated}, which defines a space of Markov chains.
\begin{definition}[IMC]
    An interval Markov chain is a tuple $\imc=(\IStates, \Plow, \Phigh, AP, L)$ where
    \begin{itemize}
        \item $\IStates$ is a finite set of states; 
        \item $\Plow: \IStates \times \IStates \to [0,1]$ is the transition interval lower-bound function, 
        where, $\forall \Istate,\! \Istate'\!\! \in \IStates$, $\Plow(\Istate, \Istate')\leq \Pr(\Istate, \Istate')$;
        \item $\Phigh: \IStates \times \IStates \to [0,1]$ is the transition interval upper-bound function, where  $\forall \Istate,\! \Istate'\!\! \in \IStates$, $\Phigh(\Istate, \Istate') \geq \Pr(\Istate,\Istate')$;
    \end{itemize}
\end{definition}

It holds that, for every $\Istate \in \IStates$, $\sum_{\Istate \in \IStates'}\Plow(\Istate,\Istate') \leq 1 \leq  \sum_{\Istate \in \IStates'}\Phigh(\Istate,\Istate')$.
Define the \textit{adversary} $\adv: \IStates \times \IStates \to [0,1]$ as a true transition probability function such that, for all $\Istate,\Istate'\in \IStates$, $\adv(\Istate,\Istate') \in [\Plow(\Istate,\Istate'), \Phigh(\Istate,\Istate')]$ and $\sum_{\Istate'\in \IStates} \adv(\Istate,\Istate') = 1$.
The set of all adversaries is denoted by $\Adv$.
Under adversary $\adv$, the IMC reduced to a Markov chain with a well-defined probability measure over its paths.

Consider a path property $\phi$.
The probability that all paths initiated at $\Istate\in\IStates$ satisfy $\phi$ is denoted by $\Pr(\Istate\models\phi)$.
When $\phi$ is expressed in probabilistic computation tree logic (PCTL) or linear temporal logic (LTL) \cite{BaierBook2008}, 
$\Pr(\Istate\models\phi)$ is equivalent to the reachability probability on an IMC that composes $\imc$ with $\phi$.  
W.L.O.G., let $\IStates_\phi \subseteq Q$ be the set of states, reaching which satisfies $\phi$. 
While the exact value of $\Pr(\Istate\models\phi)$ cannot be computed, it can be bounded, i.e., $\Pr(\Istate\models\phi) \in [\plow(\Istate), \phigh(\Istate)]$, using dynamic programming~\cite{lahijanian2015formal}.  
For the lower bound,
\begin{equation}
    \label{eq:value}
    \plow^0(\Istate) = \Indy(\Istate \in \IStates_\phi),\; 
    \plow^{k}(\Istate) = \min_{\adv\in\Adv}\!\sum_{\Istate'\in\IStates} \!\! \adv(\Istate, \Istate') \plow^{k-1}\!(\Istate').
\end{equation}
The upper bound $\phigh$ is computed by replacing the $\min$ with $\max$ operator and $\plow$ with $\phigh$.
The computation of the satisfaction bounds $\plow(\Istate), \phigh(\Istate)$ for all $\Istate\in\IStates$ is called the IMC verification procedure. 

\subsection{Problem Statements}
Verifying System~$\eqref{eq:system}$ against $\phi$ can be performed by discretizing the state space of $\eqref{eq:system}$ to build an IMC abstraction $\imc$ that soundly models $\eqref{eq:system}$, and then verifying $\imc$ against $\phi$.
Let the discretization of a compact subset of $\reals^n$ be $\IStates_X$, and let $\Istate$ refer to both an IMC state in $\IStates_X$ and a subset of $\reals^n$.
The verification results can be extended to \eqref{eq:system}, i.e., for every $x \in \Istate$, $\Pr(x \models \phi) \in [\plow(\Istate), \phigh(\Istate)]$, if the abstraction satisfies the soundness definition below as shown in \cite[Theorem 2]{jackson2021strategy}.
\begin{definition}[Abstraction Soundness]\label{def:soundness}
    An IMC abstraction $\imc$ is sound with respect to System~\eqref{eq:system} if, for all $x \in\Istate$, $\Plow(\Istate,\Istate')\leq T(\Istate'\mid x) \leq \Phigh(\Istate,\Istate')$ holds for all $\Istate\in\IStates_X$.
\end{definition}
To satisfy this definition, we assume that one of the requirements of $\phi$ is to remain within a bounded (safe) set $X \subset \reals^n$ and refer to $\reals^n\setminus X$ as an unsafe set.

Existing methods for IMC abstraction of stochastic systems are largely limited to simple dynamics -- affine in noise with unimodal or symmetric distributions, or linear dynamics.
The first problem considered here aims to establish a method that jointly addresses of these limitations.
\begin{problem}[Abstraction Construction]
    Construct a sound IMC abstraction for System~\eqref{eq:system} with a nonlinear $f$ and non-affine and non-standard $\dist$.
\end{problem}
In Section~\ref{sec:abstraction}, we propose a method that partitions the domain of $\dist$ to construct the transition bounds of the IMC which are valid for arbitrary distributions.
Solving this problem allows the application of IMC abstractions to a wider class of systems, including data-driven  systems.

The conventional approach to improving the satisfaction intervals of an IMC is to refine the discretiztion $\IStates_X$, which contributes to the state-explosion dilemma.
The next problem aims to improve the intervals on the same discretization $\IStates_X$ by leveraging the model of \eqref{eq:system}.
\begin{problem}[Verification Improvement]
    Given abstraction $\imc$ of System~\eqref{eq:system}, reduce the verification error $\phigh(\Istate) - \plow(\Istate)$ for all $\Istate\in\IStates_X$ without refining $\IStates_X$.
\end{problem}
In Section~\ref{sec:clustering}, we propose an approach based on clustering states in $\IStates_X$ that uses the structure of the transition bounds and \eqref{eq:system} to reduce the gap between $\phigh$ and $\plow$.

\begin{remark}
While we focus on IMC abstractions, the results are trivially applied to interval Markov decision process (IMDP) abstraction methods via concatenation of IMCs.
\end{remark}

\section{ABSTRACTION via NOISE PARTITIONS}\label{sec:abstraction}
The IMC abstraction for System~\eqref{eq:system} involves discretizing the continuous state-space and computing transition probability bounds between the resulting states.

\subsection{State Discretization}
Constructing a finite-state abstraction for System \eqref{eq:system} requires a bounded subset $X\subset\reals^n$.
The abstraction will be sound on $X$, but not the entire state-space as discussed in Definition~\ref{def:soundness}. 
$X$ is partitioned into a finite set of bounded and convex regions $\IStates_X$, which implies, for every $\Istate, \Istate' \in \IStates$, $\Istate \cap \Istate'$ has zero measure and $\bigcup_{\Istate\in\IStates_X} \Istate = X$.
Let $\Istate_{\neg X}=\reals^n \setminus X$. 
Then, the complete state set of the IMC is $\IStates = \IStates_X \cup \{\Istate_{\neg X}$\}.
The next IMC abstraction step computes the transition bounds between states.

\subsection{Transition Bounds with Noise Partitioning}
The definition and computation of the transition bound functions $\Plow,\Phigh$ begins with states in $\IStates_X$.  
The transitions to $\Istate_{\neg X}$ is a modified case.
The connection between System \eqref{eq:system} and the abstraction arises from the transition kernel $T$ in \eqref{eq:transition_kernel} over IMC states. 
From a given $x\in\Istate$, the transition kernel to $\Istate'$ is $T(\Istate'\mid x)$.
Finding bounds on the kernel amounts to searching over all $x\in\Istate$, i.e.,
$\min_{x\in \Istate} T(\Istate'\mid x)$, and $\max_{x\in \Istate} T(\Istate'\mid x)$.
To satisfy Definition \ref{def:soundness}, $\Plow(\Istate,\Istate')$ and $\Phigh(\Istate,\Istate')$ must bound these extrema.
For tractable evaluation of $T$ in \eqref{eq:transition_kernel} with non-standard distributions, the probability measure of $\w$ is evaluated over partitions of its domain $W$.  

\begin{definition}[Noise Partition]
    \label{def:partition}
    Let $\w(k)\in W$. 
    A noise partition set $C$ is a measure-preserving discretization of $W$, i.e., $\bigcup_{c\in C} c = W$ and $\forall c\in C$, $\sum_{c\in C}\int_c \dist d\w = \int \dist d\w = 1$.
\end{definition}
For brevity, $c$ is used in place of $\w(k) \in c$, and its probability is $\Pr(c) = \int_c \dist d\w$.  
For a given $c \in C$, the posterior of region $\Istate$ is $\Post(\Istate, c) = \{f(x,w) \mid x \in \Istate, \, w \in c\}$.
The following theorem bounds the transition kernel.
\begin{theorem}\label{thm:transition_bounds}
    Let $\Istate, \Istate' \in \IStates_X$ and 
    $C$ be a partition of $W$ according to Definition \ref{def:partition}.
    Then, the transition kernel is lower- and upper-bounded, respectively, by
    \begin{subequations}
    \begin{align}
        \min_{x_k\in \Istate} T(\Istate'\mid x_k) &\geq \sum_{c\in C} \Indy(\Post(\Istate, c) \subseteq \Istate') \Pr(c)\label{eq:lower}\\
        \max_{x_k\in \Istate} T(\Istate'\mid x_k) &\leq \sum_{c\in C} \Indy(\Post(\Istate, c) \cap \Istate' = \emptyset) \Pr(c)\label{eq:upper} \;\;\;
    \end{align}
\end{subequations}
\end{theorem}
\begin{proof}
We begin with finding the upper bound. 
Using $T$ and finding the maximizing point, 
\begin{align}
    \max_{x_k\in \Istate} T(\Istate'\mid x_k) =  \max_{x_k\in \Istate} \int \Indy(\x(k+1)\in \Istate' \mid x_k,w_k) \dist d\w \label{eq:proof1}
\end{align}
The integral is split according to the partitions in $C$,
\begin{align}
    \text{\eqref{eq:proof1}} = \max_{x_k \in \Istate} \sum_{c\in C} \int_c \Indy(\x(k+1)\in \Istate' \mid x_k\,w_k) \dist d\w  \label{eq:proof2}
\end{align}
which maintains equality due to the linearity of the integral.
The indicator function is upper-bounded by the existence of a point in the intersection of $\Post(\Istate,c)$ with $\Istate'$,
\begin{align}
    \text{\eqref{eq:proof2}} \leq\sum_{c\in C} \Indy(\Post(\Istate, c) \cap \Istate' \neq \emptyset) \Pr(c)
\end{align}
where the $\max$ operator is dropped, as $x_k$ is subsumed by $\Istate$. 
The lower-bound is similar, instead doing under-approximation by checking if $\Post(\Istate,c)\subseteq\Istate'$.
\end{proof}

The transition bounds found using Theorem \ref{thm:transition_bounds} require two components: $\Post(q,c)$ and $\Pr(c)$.  
Note that for the bounds in \eqref{eq:lower}-\eqref{eq:upper}, an over-approximation of $\Post(\Istate,c)$ can be used, which can be obtained for nonlinear systems using local linear bounds of $f(\x(k),\w(k))$~\cite{jin2021mesh,mathiesen2022safety}, discreization with Taylor model flowpipes~\cite{chen2013flow}, or mixed-monotone maps~\cite{dutreix2018efficient} depending on the knowledge of System \eqref{eq:system}.
$\Pr(c)$ can be computed analytically for distribution-dependent soundness guarantees, or statistically for sampling-based guarantees\cite{badings2023robust}.
The next section discusses how partitions are selected to optimize the bounds in Theorem~\ref{thm:transition_bounds}.

To complete the abstraction, transitions to the unsafe state $\Istate_{\neg X}$ are defined using the following corollary.
\begin{corollary}[Unsafe State Transitions]
    For every state $\Istate\in\IStates_X$, the transition bounds to $\Istate_{\neg X}$ are 
        $\Plow(\Istate, \Istate_{\neg X}) = 1 - \max_{x_k \in q} T(X \mid x_k)$ and 
        $\Phigh(\Istate, \Istate_{\neg X}) = 1 - \min_{x_k \in q} T(X \mid x_k)$.
    Additionally, the transition bounds between $\Istate_{\neg X}$ and itself are $\Plow(\Istate_{\neg X}, \Istate_{\neg X}) = \Phigh(\Istate_{\neg X}, \Istate_{\neg X}) = 1$.
\end{corollary}

\begin{remark}
    Theorem \ref{thm:transition_bounds} can be applied to general (non-probabilistic) uncertainty sets by interpreting $\Pr(c)$ as a deterministic indicator function.
    For example, for the bounded uncertainty set $W$, choose $c=W$ so $\Pr(c)=1$, and $\Pr(c')=0$ for every other $c'\in C$.
    Effectively, using Theorem \ref{thm:transition_bounds} in this case results in a non-deterministic transition system.
\end{remark}

\section{OPTIMAL PARTITIONS}\label{sec:partitions}
The transition bounds in Theorem \ref{thm:transition_bounds} return valid bounds for any choice of partition, and $C$ can differ between \eqref{eq:lower} and \eqref{eq:upper}.
However, haphazard partitions can result in the trivial transition probability interval $[0,1]$.
The optimal noise partitions  minimize the distance between the transition bounds, i.e.,
given $\Istate,\Istate' \in \IStates_X$,
\begin{subequations}\label{eq:optimal}
\begin{align}
    \Clower &= \argmax_{C} \sum_{c\in C}  \Indy(\Post(\Istate, c) \subseteq \Istate') \Pr(c), \label{eq:opt-lower}\\
    \Cupper &= \argmin_{C} \sum_{c\in C} \Indy(\Post(\Istate, c) \cap \Istate' = \emptyset) \Pr(c), \label{eq:opt-upper}
\end{align}
\end{subequations}
Hence, noise partitions can be optimized for each pair $(\Istate,\Istate')$.

To begin the analysis on these partitions, we assume component-wise noise as defined below.
\begin{definition}[Component-wise Noise]\label{def:components}
    For $i \in \{1,\ldots, n\}$, let $M^i \in \{0,1\}^{n \times n}$ be a matrix whose $i,i$ element is one and all the other elements are zeros.  
    Then, noise $\w(k)\in W\subset\reals^n$ is called \emph{component-wise} if $M^i f(\x(k), \w(k)) = f(\x(k), M^i \, \w(k))$. 
\end{definition}
In other words, the noise vector shares the size of $\x(k)$, and each component $\wi(k)$ only affects $\xi(k+1)$, which admits (but is not limited to) affine and multiplicative noise (see Example \ref{ex:multi}). 
Definition \ref{def:components} does not preclude the noise from being correlated, and the following remark concerns more general structures.
\begin{remark}
    Relaxing the component-wise noise in Definition \ref{def:components} relies on effective noise.
    For example, consider the term $B\w(k)$ with $\w(k)\in\reals^{n_\w}$ and $B\in\reals^{n\times n_\w}$.
    The effective noise terms are $b_i\w(k)$, where $b_i$ is the $i$-th row of $B$, each of which affects $\xi(k+1)$.
\end{remark}

Assuming the noise satisfies Definition \ref{def:components}, the next step is examining how $\Post(\Istate,c)$ changes with respect to $c$.
Intuitively, if increasing $\wi_k$ always increases (or decreases) $\xi(k+1)$, then there can be partitions $C$ where non-empty intersections $\Post(\Istate, c)$ between $\Istate'$ are induced.
This is due to the monotonicity of the system with respect to the noise, which is formally defined below.
\begin{definition}[Noise Monotonicity]
    Monotonicity is the condition $\wi_{k}>\wi_{j} \implies \xi_{k+1}\geq\xi_{j+1}$ or $\xi_{k+1}\leq\xi_{j+1}$ with $\xi_k=\xi_j$. 
    System~\eqref{eq:system} is \emph{monotonic} with respect to $\w(k)$ if each $\xi(k+1)$ is monotonic with respect to $\wi(k)$.
\end{definition}
\begin{example}\label{ex:multi}
    Consider the system $\x(k+1) = f(x(k))\odot \w(k)$, where each $\wi(k)\geq 0$ and $\odot$ is the element-wise product.
    Then the noise acts component-wise, and the system is monotonic with respect to $\w(k)$.
\end{example}
Hitherto, we have made no assumptions about the convexity of $\Post(\Istate,c)$.
Let the $i$ component of a set in $\reals^n$ refer to its projection on the $i$-th unit axis.
The following theorem discusses non-convexity in terms of discontinuities (or holes) in each component of $\Post(\Istate,c)$.
The theorem bounds the sizes of $\Clower$ and $\Cupper$ for a system that is monotonic with respect to uncorrelated noise $\w(k)$.
\begin{theorem}[Partition Size]\label{thm:partitions}
    Let $\Istate,\Istate' \subset \reals^n$ be bounded and convex, each component of $\Post(\Istate,c)$ contain at most $d$ discontinuities, and $\Clower$, $\Cupper$ be as in \eqref{eq:opt-lower}-\eqref{eq:opt-upper}.
    If System~\eqref{eq:system} is monotonic with respect to $\w(k)$, and the components of $\w(k)$ are uncorrelated, then $|\Cupper|$ and $|\Clower|$ is at most $(3+2d)n$.
\end{theorem}
\begin{proof}
    The proof is provided for the upper-bound partition \eqref{eq:opt-upper}.
    The lower-bound is the same, but instead uses $\Post(\Istate,c)\subseteq\Istate'$. 

    $\Cupper$ is found my choosing the constraint set that satisfies to \eqref{eq:opt-upper}.
    As the noise is uncorrelated, it is sufficient to minimize the area intersection of $\Post(\Istate,c)$ with $\Istate'$ to minimize \eqref{eq:opt-upper}.
    Let $\Istate$, $\Istate'$ be convex and bounded, let $\Post(\Istate,c)$ contain at most $d$ discontinuities for any choice of $c$, and let System \eqref{eq:system} be monotonic with respect to $\w(k)$. 

    First, consider $d=0$, so $\Post(\Istate,c)$ is convex for a given $c$ in all components. 
    Then, $\Post(\Istate,c)\cap\Istate'$ must be convex. 
    As $\Istate'$ is bounded, $\Post(\Istate,c)\cap\Istate'$ is also bounded. 
    For each $i$, due to the monotonicity of System \eqref{eq:system}, at most 3 partitions of $W^i$ are needed to induce the minimum intersection with $\Istatei$ due to the convexity and boundedness of $\Istatei$.
    Thus, $\Cupper$ consists of $3n$ partitions at most when $d=0$.

    Next, consider $d>0$ for a component of $\Post(\Istate, c)$ and begin with $\Cupper$ as found above. 
    The intersection $\Post(\Istate,c)\cap\Istate'$ is possibly non-convex due the projection of discontinuities of $\Post^i(\Istate,c)$. 
    For each discontinuity, only two additional partitions are needed to induce the minimum intersection with $\Istatei$ due to the monotonicity of System \eqref{eq:system}.
    This is repeated for each component in $[1,n]$ for the result $|\Cupper|\leq(3+2d)n$. 
    Repeating this procedure for the lower bound yields the same number of partitions in $\Clower$. 
\end{proof}
Theorem \ref{thm:partitions} shows that the sizes of $\Clower$ and $\Cupper$ are bounded, but it leaves them unspecified.
The following corollaries specify the partitions for affine and multiplicative noise in the case $\Post(\Istate,c)$ is convex.
To facilitate this, let $\Post_f(\Istate) = \{f(x) \mid x\in\Istate\}$ be the $f$-dependent posterior.
Each corollary specifies the paritioning of $W^i$ into three intervals, i.e. $\{[-\infty,\epsilon_1], [\epsilon_1,\epsilon_2], [\epsilon_2, \infty]\}\subset\Cupper$ and $\{[-\infty,\epsilon_3], [\epsilon_3,\epsilon_4], [\epsilon_4, \infty]\}\subset\Clower$.

\begin{figure}[t]
    \centering
    \begin{tikzpicture}[scale=1.5]
    \draw[ultra thick] (-3,0) -- (1,0);
    \draw[ultra thick] (0.0,-0.15) -- (1.5, -0.15);
    \draw[dashed] (-3.0, -0.15) -- (-1.5, -0.15);
    \node[below] at (-1.0, 0) {$\Istate'$};
    \node[below] at (0.5,-0.15) {$\Post_f(\Istate)$};
    \node[below] at (-2.5,-0.15) {$\Post_f(\Istate)$ Shadow};
    \node[left] at (-3.0,0) {$A$};
    \node[right] at (1.0,0) {$B$};
    \node[left] at (0.0, -0.15) {$C$};
    \node[right] at (1.5, -0.15) {$D$};

    \draw[line width=0.4] (-3.0, -0.20) -- (-3.0, 0.25);
    \draw[line width=0.4] (1.5, -0.20) -- (1.5, 0.25);
    \draw[<-, line width=0.4] (-3.0, 0.15) -- (1.5, 0.15);
    \node[above] at (-0.75, 0.15) {$\epsilon_2$};

    \draw[line width=0.4] (0.0, -0.20) -- (0.0, 0.35);
    \draw[line width=0.4] (1.0, -0.4) -- (1.0, 0.35);
    \draw[->, line width=0.4] (0.0, 0.25) -- (1.0, 0.25);
    \node[above] at (0.5, 0.25) {$\epsilon_1$};

    \draw[line width=0.4] (1.5, 0.0) -- (1.5, -0.75);
    \draw[<-, line width=0.4] (1.0, -0.25) -- (1.5, -0.25); 
    \node[below] at (1.25, -0.25) {$\epsilon_3$}; 

    \draw[line width=0.4] (-1.5, 0.0) -- (-1.5, -0.75);
    \draw[<-, line width=0.4] (-1.5, -0.65) -- (1.5, -0.65); 
    \node[above] at (-0.25, -0.65) {$\epsilon_4$}; 
\end{tikzpicture}
    \caption{Distances between a component of $\Post_f(\Istate)$ and $\Istate'$ used to find optimal noise partitions.}
    \label{fig:distances}
\end{figure}
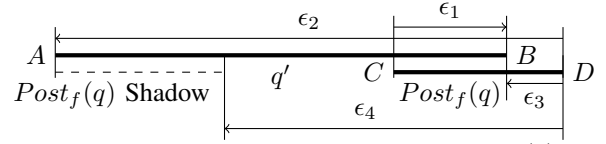

\begin{corollary}[Partitioning for Affine Noise]\label{cor:affine}
    Assume System~\eqref{eq:system} satisfies the requirements of Theorem~\ref{thm:partitions} and has affine noise, i.e. $f(\x(k)) + \w(k)$. 
    Let the $i$-th component vertices be $[A,B]$ for the target $\Istate'$, $[C,D]$ for the $f$-dependent posterior, and let $W = D-C$. 
    Let Case 1 be $W>B-A$, Case 2 be $D>B$, and Case 3 be $C<A$.
    Then $\epsilon_1 = A - D$, $\epsilon_2 = B - C$ in all cases. 
    For Case 1, $\epsilon_3 = 0$ and $\epsilon_4 = 0$.
    If Case 2 is true, $\epsilon_3 = A+W-D$, and if Case 3 is true, $\epsilon_4 = B-W-C$.
    Otherwise, $\epsilon_3 = A-C$ and $\epsilon_4 = B-D$.
\end{corollary}
\begin{proof}
Figure \ref{fig:distances} illustrates Case 2 of the relative geometry of a component of $\Post_f(\Istate)$ and a target state $\Istate'$.
The lower-bound \eqref{eq:lower} is maximized when $\epsilon_3 < \w(k) \leq \epsilon_4$ as this is the largest interval that results in an intersection.
The upper-bound \eqref{eq:upper} is minimized when $\w(k) < \epsilon_1$ or $\epsilon_2 < \w(k)$ as the intersection is zero outside of this interval.
The partitions for other relative component geometries of $\Post_f(\Istate)$ and $\Istate'$ are found similarly.
\end{proof}

\begin{corollary}[Partitioning for Multiplicative Noise]
    Assume System~\eqref{eq:system} satisfies the requirements of Theorem~\ref{thm:partitions} and has (w.l.o.g.) positive multiplicative noise, i.e., Example \ref{ex:multi}. 
    Let the component vertices be $[A,B]$ for the target and $[C,D]$ for the $f$-dependent posterior, and let $A,B,C,D>0$.
    Then, $\epsilon_1 = A/D,~ 
        \epsilon_2 = B/C,~
        \epsilon_3 = A/C,$ and 
        $\epsilon_4 = B/D$.
\end{corollary}
\begin{proof}
The proof is based on the relative geometry in Figure \ref{fig:distances} similar to that of Corollary \ref{cor:affine}.  
The lower-bound is maximized when both $D\w(k) < B$ and $C\w(k) > A$, so $A/C < \w(k) \leq B/D$. 
Note that if $A/C>B/D$, then the CDF evaluates to zero, so the partition set is trivial.
Likewise, the upper-bound is minimized when $C\w(k) > B$ or $D\w(k) < A$, so $\w(k) < A/D$ or $\w(k) > B/C$.
\end{proof}

The case studies in Section \ref{sec:examples} demonstrate that using these optimal partitions find accurate abstractions for nonlinear systems with non-standard noise, even when compared to specialized methods.

\section{STATE CLUSTERING}\label{sec:clustering}
Improving the satisfaction intervals of the IMC directly impacts the guarantees on \eqref{eq:system}, but relying solely on refining the space discretization can lead to an explosion in the number of states.
We propose a novel method based on clustering the states of the IMC to improve the satisfaction intervals without refinement.

Consider state $\Istate\in\IStates_X$ and its possible successor states $\IStates'$.  
By the structure of \eqref{eq:lower}, $\Plow(\Istate,\Istate')$ grows as the size of $\Istate'$ increases as it depends on $\Post(\Istate,c)\subset\Istate'$ being true.
Algorithm \ref{alg:clustering} is based on this principle. 
The sorting of $\IStates_X$ on Line \ref{alg:sort} makes the algorithm start with states with large $\plow$, as its successor states have larger $\plow$, making improvement more likely. 
Then, the states in $\tilde\IStates$ are clustered (or merged) into a single state $\tilde\Istate$.  
The transition interval to $\tilde\Istate$ is computed using Theorem \ref{thm:transition_bounds}, and the satisfaction intervals are recalculated.
This procedure can be repeated until no improvements in $\plow$ or $\phigh$ are 
realized.

Our evaluations suggests the efficacy of this algorithm for a coarse abstractions.
An interesting question, however, arises with this approach: what is the optimal set of states to cluster?  
We leave this question for future work.  

\begin{algorithm}[t]
    \caption{Clustering-based IMC improvement}\label{alg:clustering}
    \begin{algorithmic}[1]
        \Require IMC $\imc$, verification results $\plow$, $\phigh$
        \State $\IStates_X\gets$ sort by $\plow(\Istate)$ in descending order \label{alg:sort}
        \ForEach{$\Istate\in\IStates_X$}
            \State $\tilde \Istate\gets$ cluster $\tilde\IStates\subset\IStates'$ into a single state \label{alg:cluster}
            \State $\plow(\tilde\Istate) \gets \min\limits_{\Istate' \in \tilde\IStates} \plow(\Istate')$ \label{alg:minimum}
            \State $\Plow(\Istate, \tilde \Istate),\Phigh(\Istate, \tilde\Istate)\gets$ Theorem \ref{thm:transition_bounds}
            \State $\plow_{new}(\Istate)\gets \min\limits_{\adv\in\Adv}  \sum\limits_{\Istate' \in \IStates'\setminus \tilde\IStates}\adv(\Istate,\Istate')\plow(\Istate') + \adv(\Istate,\tilde\Istate)\plow(\tilde\Istate)$ \label{alg:lower}
            \State $\phigh_{new}(\Istate)\gets$ similarly with \eqref{eq:value}
            \If{$\plow_{new}(\Istate) > \plow(\Istate)$ or $\phigh_{new}(\Istate) < \phigh(\Istate)$}
                \State Save these values into $\plow$, $\phigh$ \label{alg:improve} 
            \EndIf
        \EndFor
        \State \Return Improved intervals $\plow$, $\phigh$ 
    \end{algorithmic}
\end{algorithm}

\section{Evaluations}\label{sec:examples}
We evaluate the proposed methods on linear, non-linear, and data-driven systems, with different noise distributions and structures.
All systems are verified against the PCTL specification 
$\phi$ that states
``the probability of reaching goal $G$ within $k$ steps while avoiding obstacles $O$ is $\geq 0.9$.''
Figures include classifications of IMC states satisfying ($\models\phi$) or violating ($\not\models\phi$) the specification, and possibly either ($?\phi$).
The value of $k$ is infinity unless otherwise noted.

\subsubsection{Linear System Comparison}\label{sec:linear1}
We first compare our approach developed for general dynamics and noise distributions against the direct-search method in \cite{dutreix2018efficient} that is tuned for specific dynamics and noise models.
The considered system is linear with additive truncated Gaussian noise taken from \cite{dutreix2018efficient}.
Figure~\ref{fig:linear1} compares the results on abstractions found using these two methods.
Our method provides near-identical classification results with average and maximum differences of $8\times 10^{-4}$ and $0.02$, respectively, in the lower-bound satisfaction probability, and took approximately 90 seconds to compute.
In this example, the set-based criteria of Theorem~\ref{thm:transition_bounds} are sufficient to provide an accurate 
abstraction.

\newcommand{\figwidth}{0.40\linewidth}
\setlength{\abovecaptionskip}{2mm}
\setlength{\belowcaptionskip}{0mm}
\newcommand{\subcapaboveskip}{\vspace{-6mm}}

\begin{figure}
    \centering
    \begin{subfigure}{0.1\linewidth}
        \centering
        \raisebox{27px}{\scalebox{.8}{
        \begin{tikzpicture}
            \definecolor{c1}{RGB}{0,175,245}
            \definecolor{c2}{RGB}{214,250,255}
            \definecolor{c3}{RGB}{213,86,114}
            
            \filldraw[fill=c1, draw=black] (0,0) rectangle (0.5,0.5);
            \node[anchor=north] at (0.25, 1.1) {$\models\phi$};
            \filldraw[fill=c2, draw=black] (0,1.1) rectangle (0.5,1.6);
            \node[anchor=north] at (0.25,2.2) {$ ? \phi$};
            \filldraw[fill=c3, draw=black] (0,2.2) rectangle (0.5,2.7);
            \node[anchor=north] at (0.25 ,3.3) {$\not\models\phi$};
        \end{tikzpicture}
        }}
    \end{subfigure}
    \begin{subfigure}{\figwidth}
        \includegraphics[trim=300px 0 300px 0, clip, width=\linewidth]{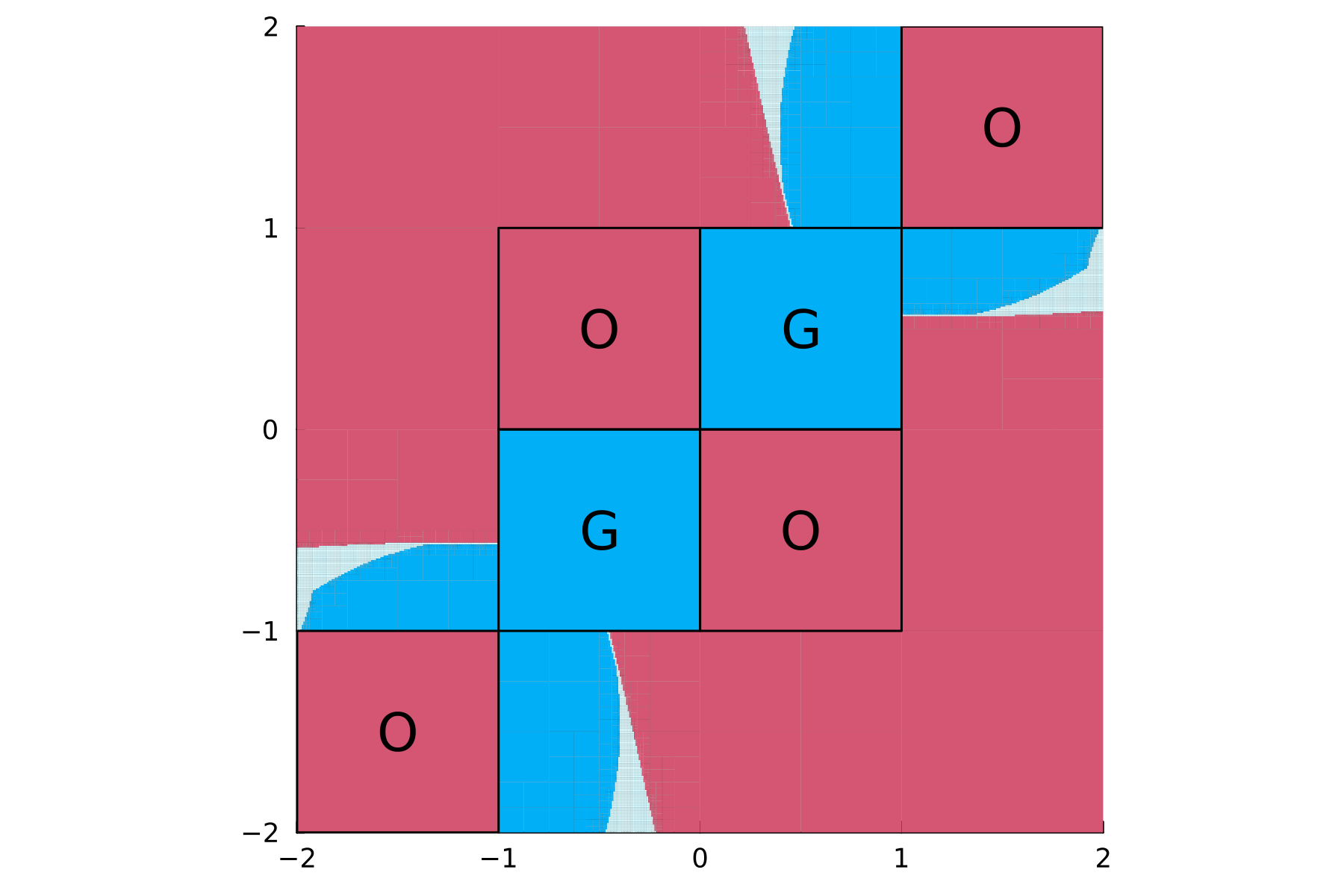}
        \subcapaboveskip
        \caption{Our method}
    \end{subfigure}
    \begin{subfigure}{\figwidth}
        \includegraphics[trim=300px 0 300px 0, clip, width=\linewidth]{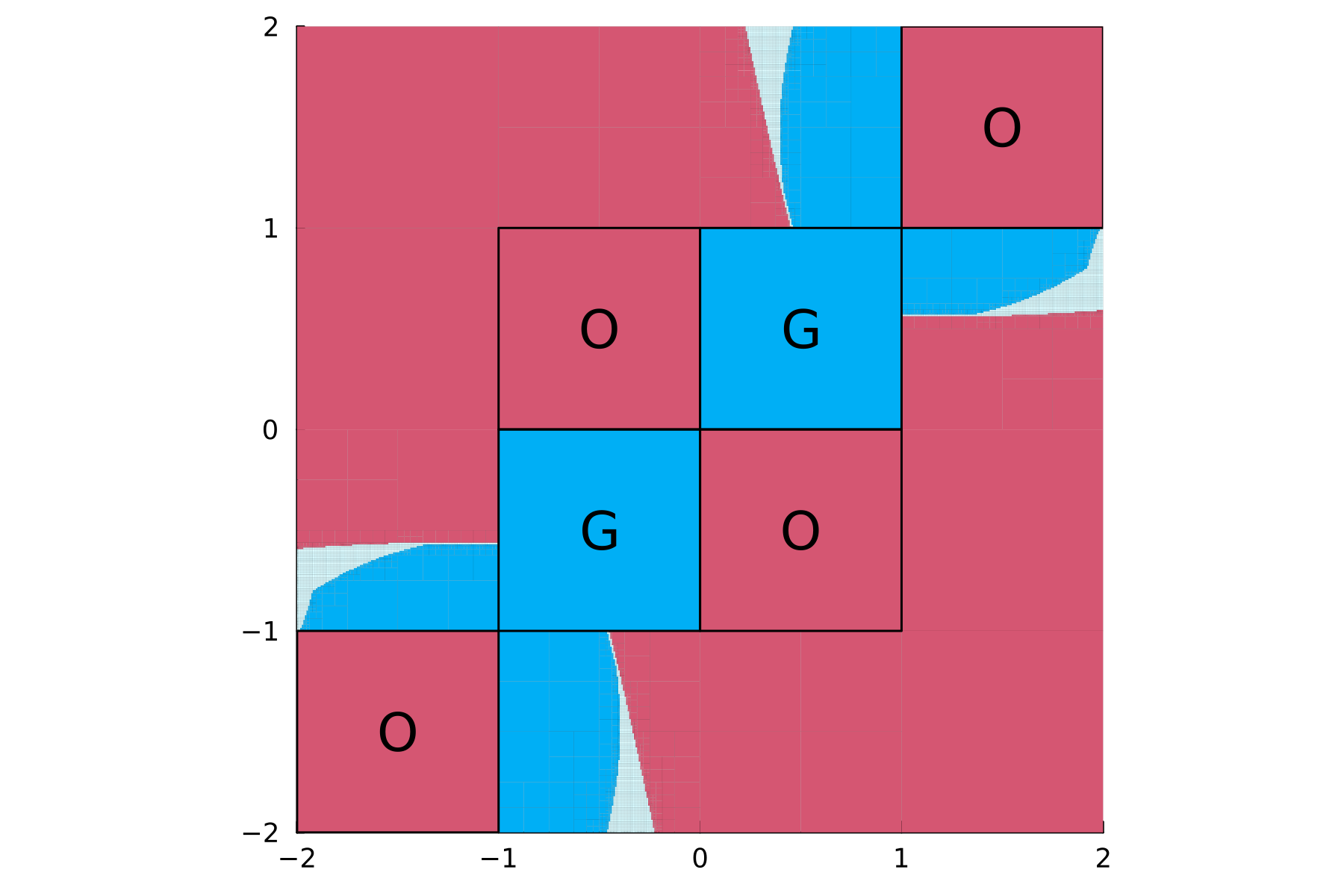}
        \subcapaboveskip
        \caption{Method in \cite{dutreix2022abstraction}}
    \end{subfigure}
    \caption{Comparison of verification results.
    }\label{fig:linear1}
    \vspace{-3mm}
\end{figure}
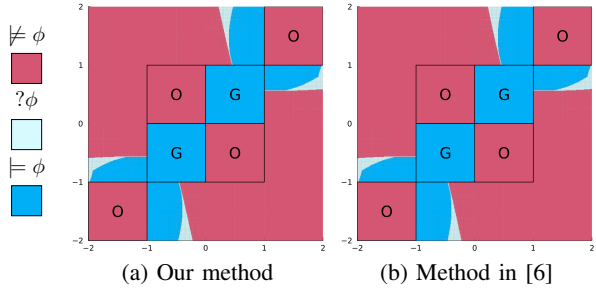

\subsubsection{Multiplicative Noise}
To show the flexibility of the approach, we consider a system with multiplicative noise, which existing abstraction approaches cannot explicitly handle, to the best of our knowledge. 
The dynamics are $\x(k+1)=A\x(k)\odot\w(k)$, where the 1st and 2nd rows of $A$ are $(0.7,0.1)$ and $(0.1,0.8)$, $\wi(k)\sim \mathcal{N}(1, 0.1)$ bounded in $[0.9, 1.1]$ for $i\in\{1,2\}$.  
Figure~\ref{fig:multiplicative1} shows the results of the procedure, which took two minutes to compute, including 1000 Monte Carlo validation trajectories with mean paths and 1-sigma confidence ellipsoids.

\begin{figure}
    \centering
    \begin{subfigure}{\figwidth}
        \includegraphics[trim=300px 0 300px 0, clip, width=\linewidth]{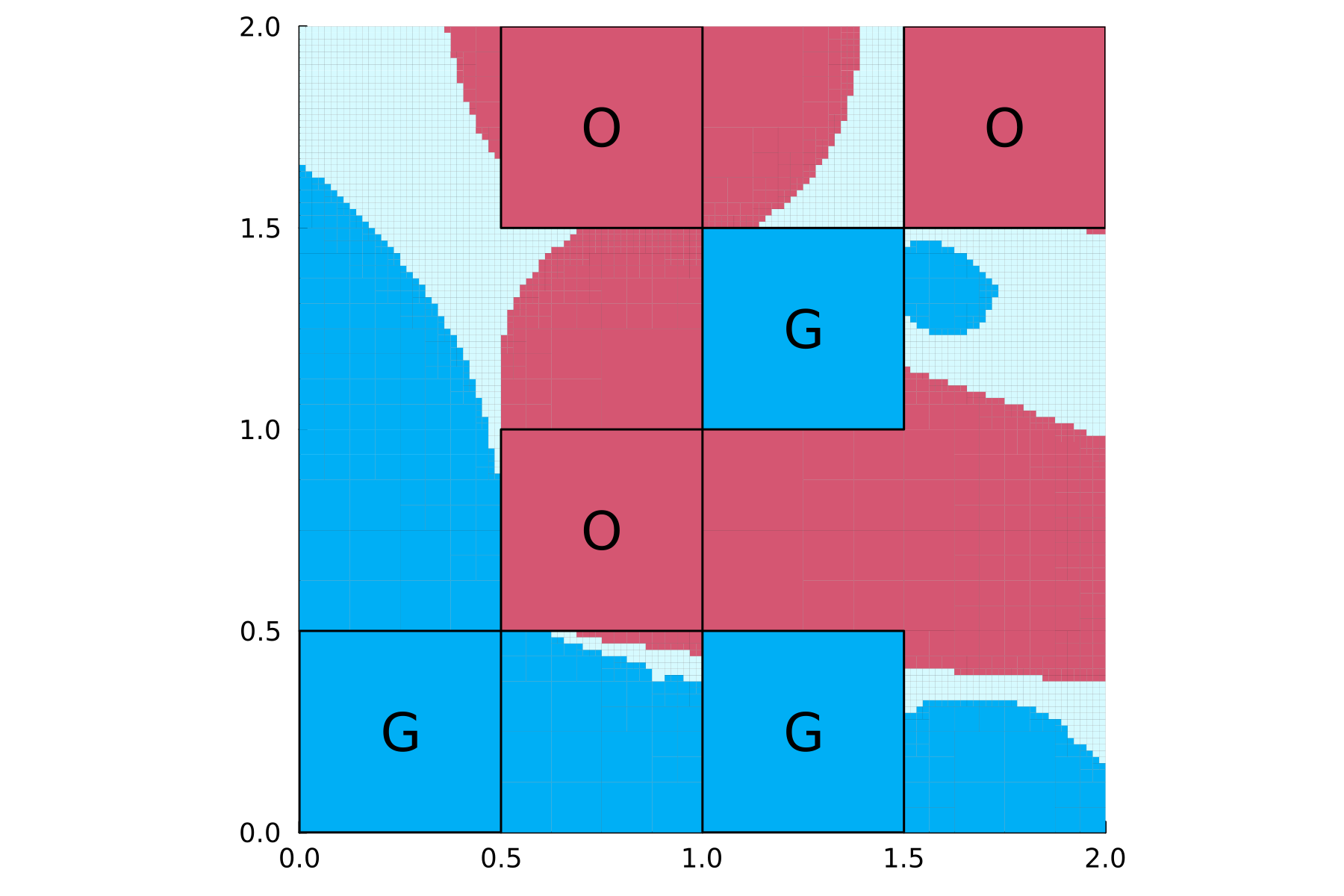}
        \subcapaboveskip
        \caption{Verification}
    \end{subfigure}
    \begin{subfigure}{\figwidth}
        \includegraphics[trim=300px 0 300px 0, clip, width=\linewidth]{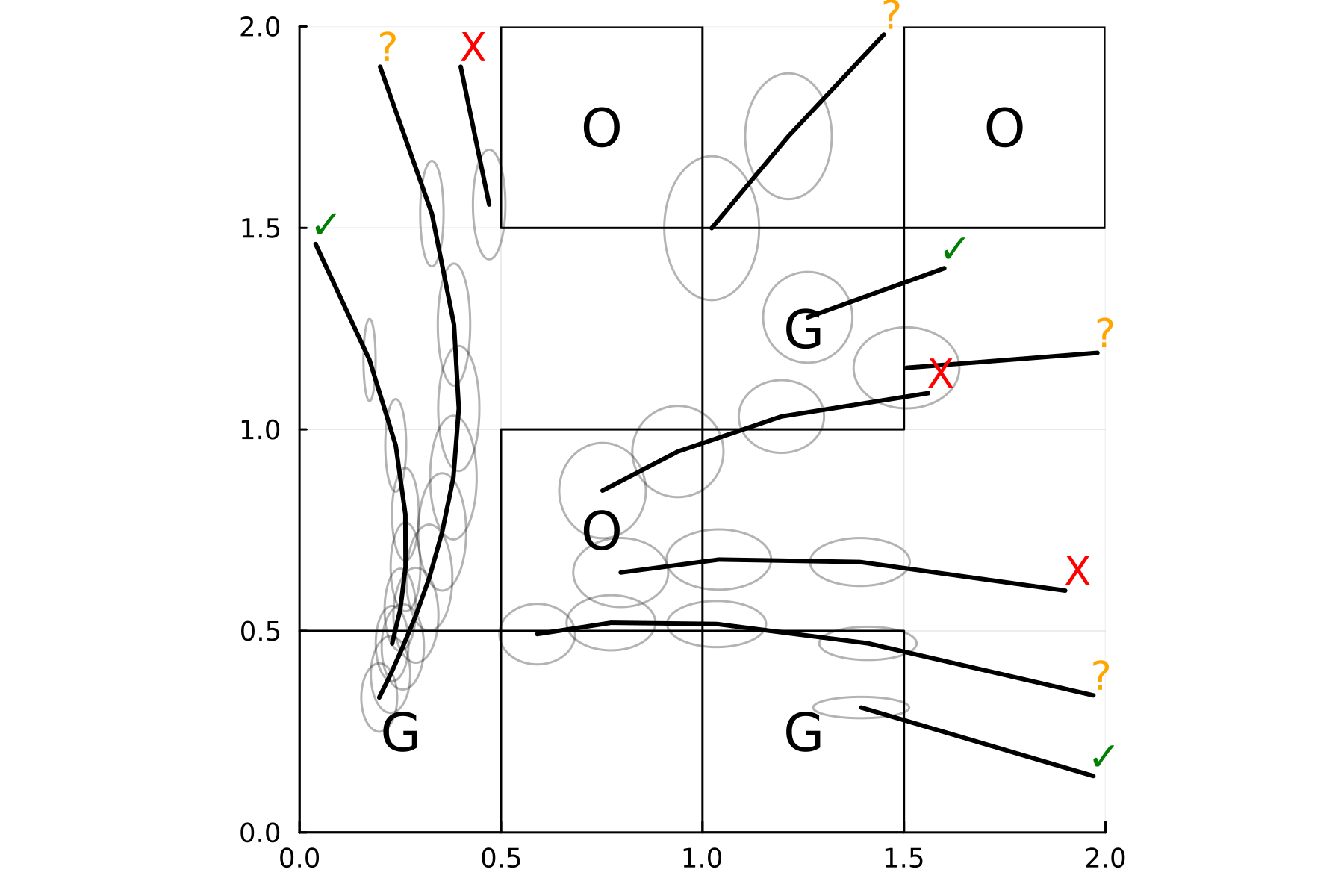}
        \subcapaboveskip
        \caption{Sampled Trajectories}
    \end{subfigure}
    \caption{Verification of the system with multiplicative noise.}\label{fig:multiplicative1}
\end{figure}

\subsubsection{Data-driven Verification}\label{sec:data-driven}
The linear system from the comparison example with $\wi(k)\sim\mathcal{N}(0,0.1^2)$ on each component is learned via Gaussian process (GP) regression with 200 data points.
The transition bounds must take into account both the resulting uncertainty from the learning procedure and the inherent system noise.
The efficacy of Algorithm \ref{alg:clustering} is demonstrated on the initial discretization.
Figures \ref{fig:cluster1} and \ref{fig:cluster2} show additional satisfying and violating states are identified without refinement.  
The lower-bound satisfaction was improved in eight states, with a 10\% (absolute) average increase.
Figures \ref{fig:learned1} and \ref{fig:learned2} compare the classifications of the learned system using Theorem \ref{thm:transition_bounds}, employing refinement, and the classifications using ideal knowledge.

\begin{figure}
    \centering
    \begin{subfigure}{\figwidth}
        \includegraphics[trim=300px 0 300px 0, clip, width=\linewidth]{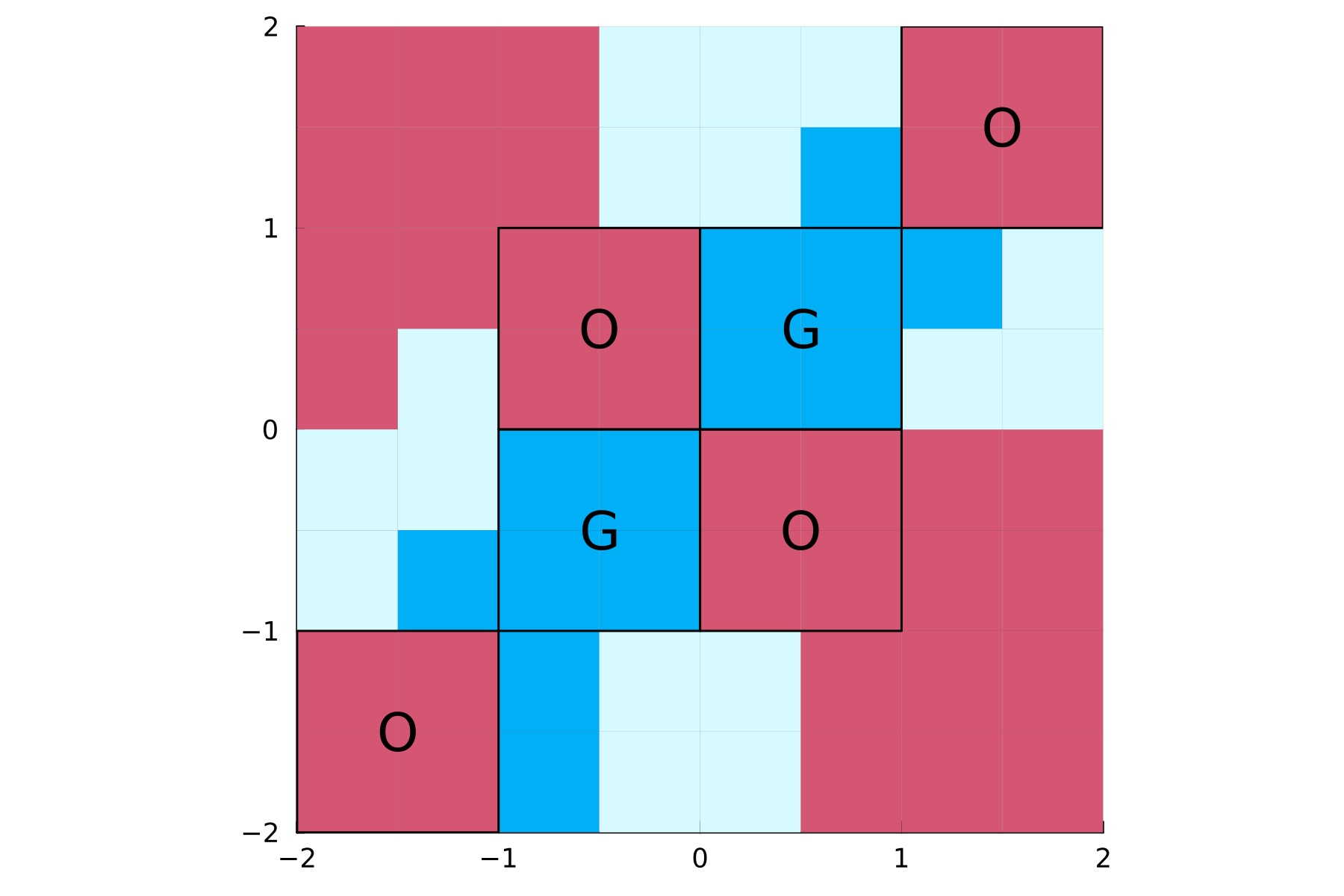}
        \subcapaboveskip
        \caption{Initial}\label{fig:cluster1}
    \end{subfigure}
    \begin{subfigure}{\figwidth}
        \includegraphics[trim=300px 0 300px 0, clip, width=\linewidth]{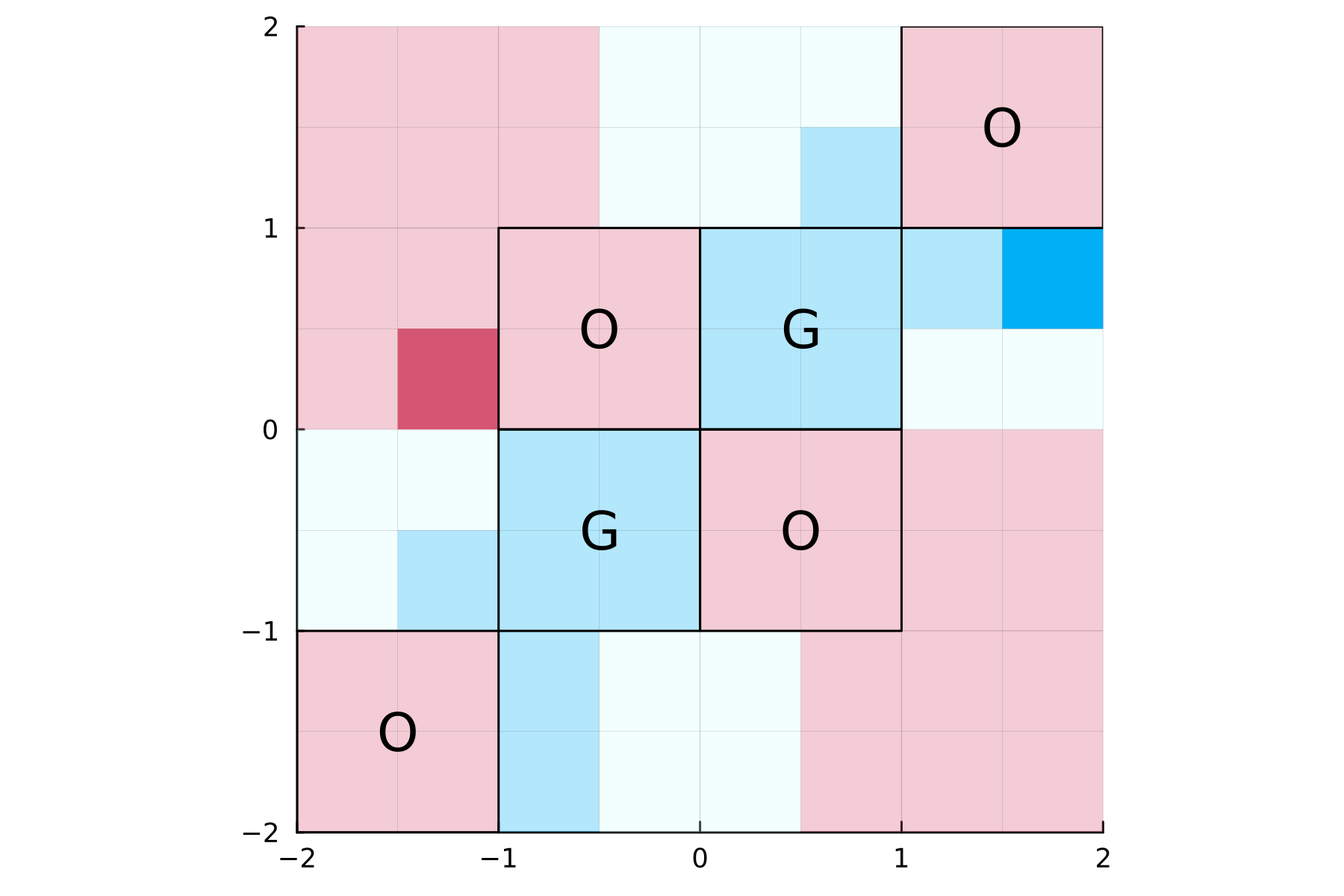}
        \subcapaboveskip
        \caption{Cluster (1 step)}\label{fig:cluster2}
    \end{subfigure}
    \begin{subfigure}{\figwidth}
        \includegraphics[trim=300px 0 300px 0, clip, width=\linewidth]{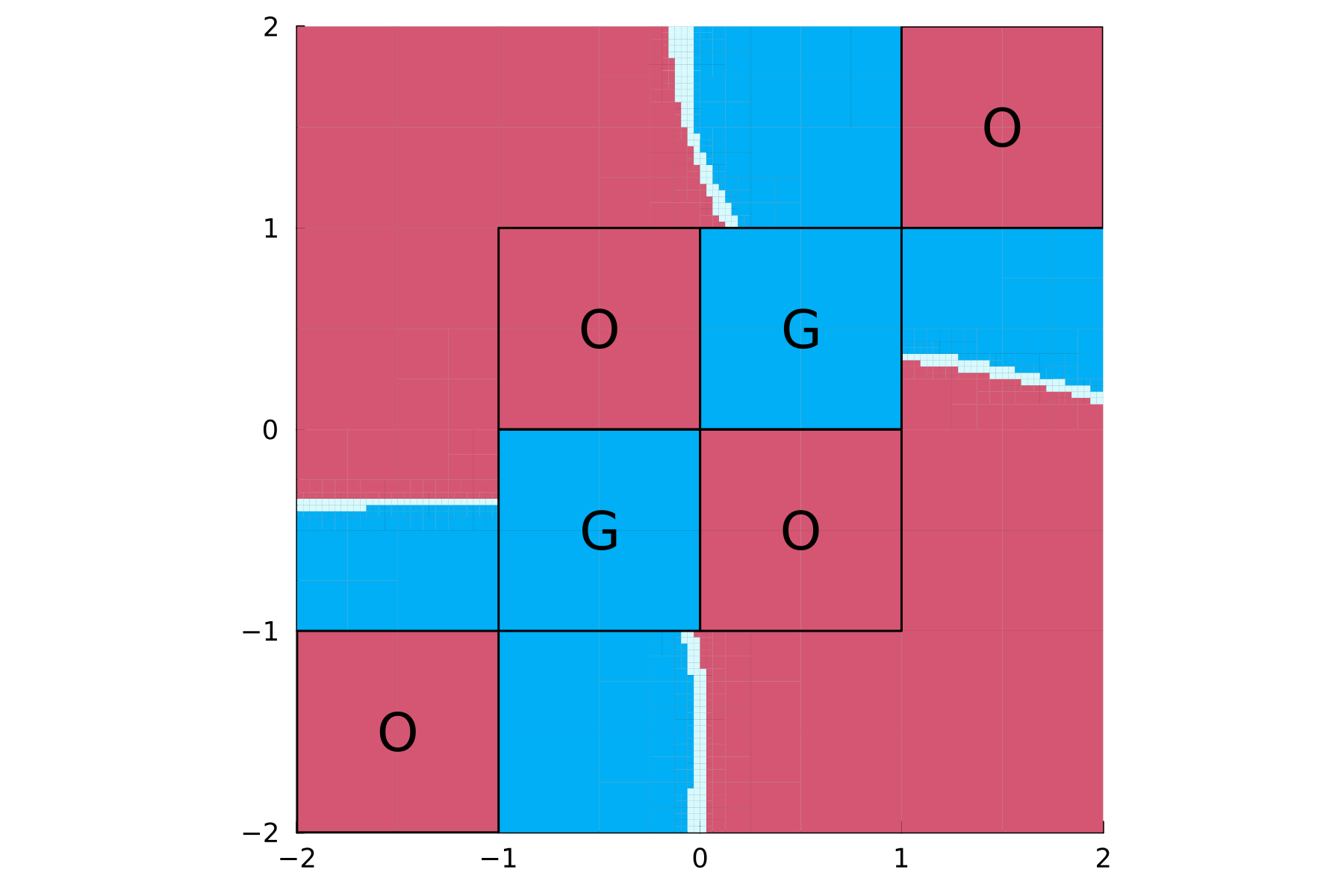}
        \subcapaboveskip
        \caption{Final (GP)}\label{fig:learned1}
    \end{subfigure}
    \begin{subfigure}{\figwidth}
        \includegraphics[trim=300px 0 300px 0, clip, width=\linewidth]{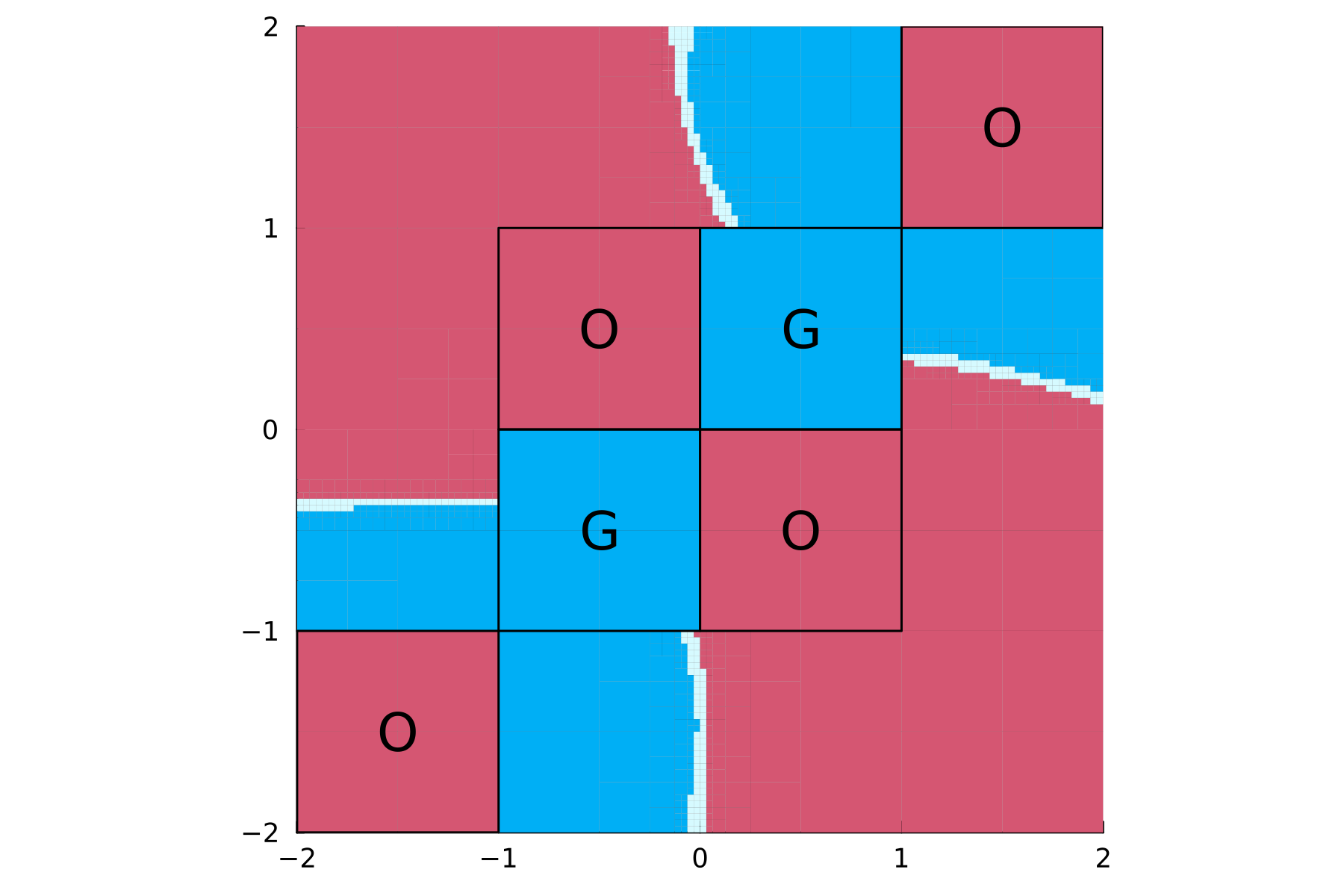}
        \subcapaboveskip
        \caption{Final (Known)}\label{fig:learned2}
    \end{subfigure}
    \caption{Verification of the learned linear system.}
\end{figure}

\subsubsection{Duffing Oscillator}
Next, we consider the nonlinear Duffing oscillator that has complex harmonic motion and chaotic behavior.
The continuous-time dynamics are 
$\ddot{x} + \delta \dot{x} + \alpha x + \beta x^3 = \gamma \cos(\omega t),$
where, $\delta=0.3$, $\alpha=-1.0$, $\beta=1.0$, $\gamma=0.37$, and $\omega=1.2$.
This system is discretized over the time-span $[0,0.5]$, after which the forcing function is reset and noise $\wi(k)\sim\mathcal{N}(0.1, 0.01^2)$ is drawn.
Taylor models were used to over-approximate the $\Post$ of each discrete region~\cite{JuliaReach19}.
The abstraction and verification for $k=10$, shown in Figure \ref{fig:duffing1}, took 4.5 hours to compute.
The trajectories show the means of 1000 sampled path with 2-sigma confidence ellipsoids at each discrete endpoint, and the initial classification according to the IMC verification.

\begin{figure}
    \centering
    \begin{subfigure}{\figwidth}
        \includegraphics[trim=300px 0 300px 0, clip, width=\linewidth]{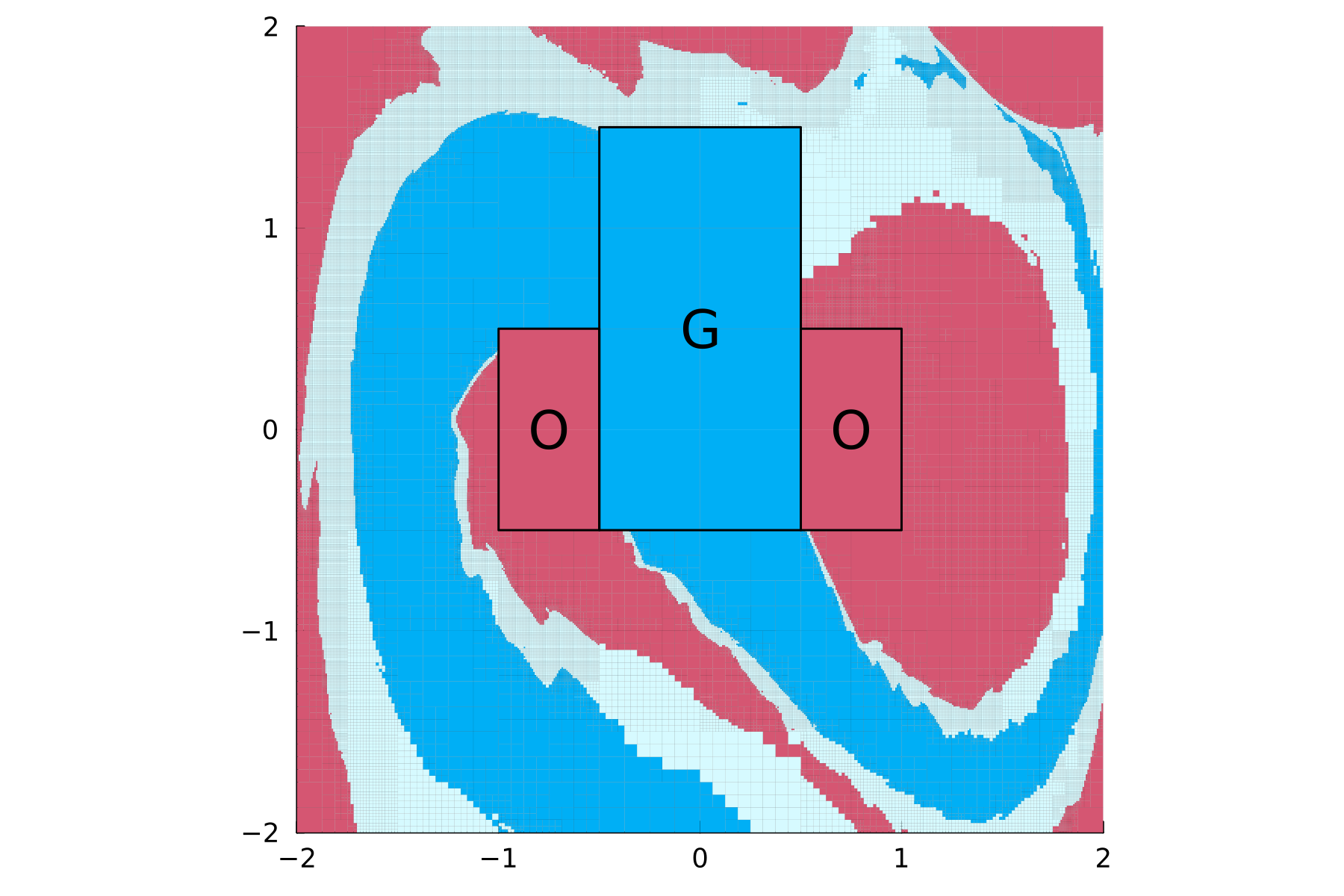}
        \subcapaboveskip
        \caption{Classifications}
    \end{subfigure}
    \begin{subfigure}{\figwidth}
        \includegraphics[trim=300px 0 300px 0, clip, width=\linewidth]{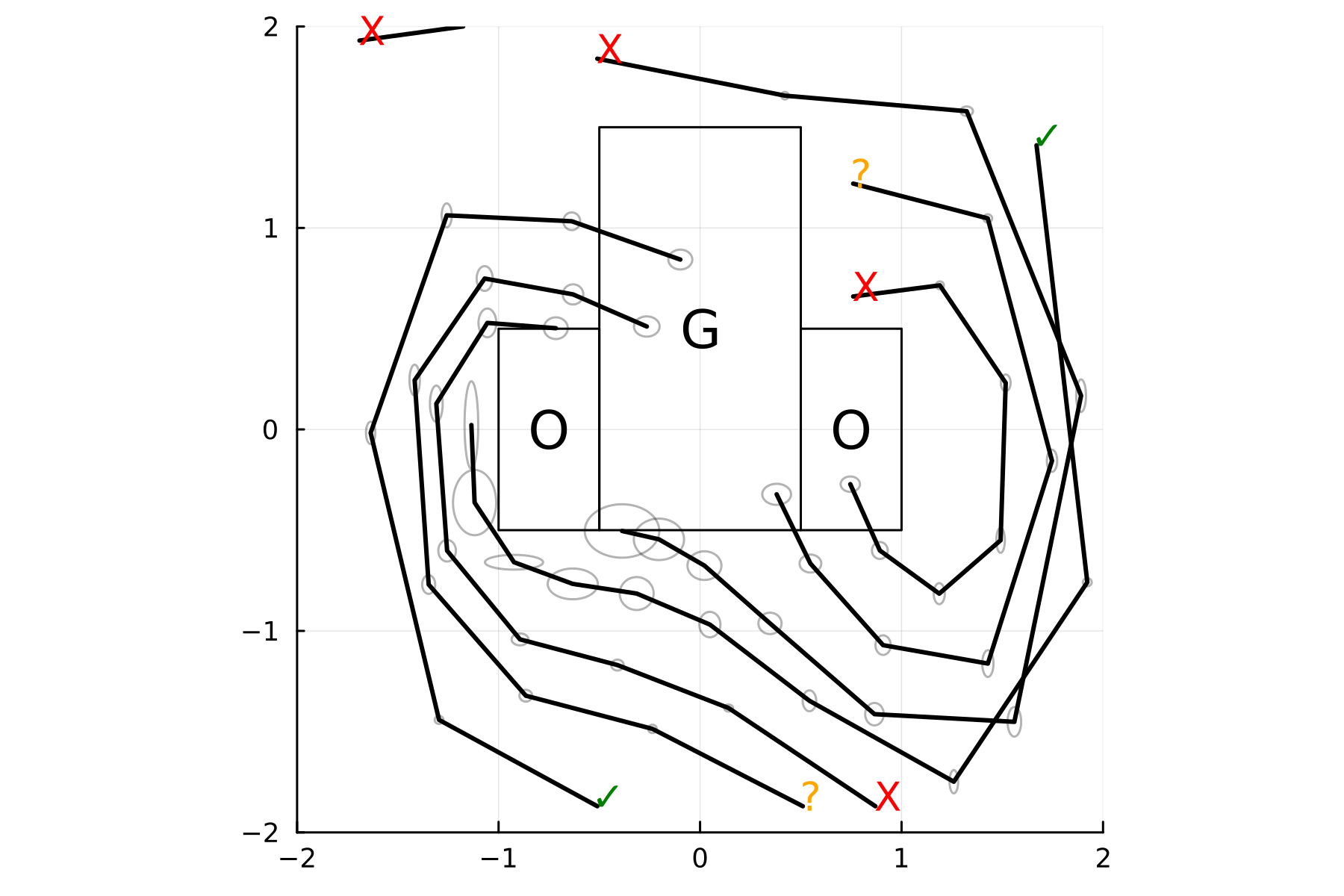}
        \subcapaboveskip
        \caption{Sampled Trajectories}
    \end{subfigure}
    \caption{Verification of the noisy Duffing oscillator.}\label{fig:duffing1}
\end{figure}

\subsubsection{Dubin's Aircraft with Mixture}
A 3D nonlinear Dubin's car model~\cite{Owen2015} with a constant right-turn control input is verified with mixture noise model consisting of $\textsc{Uniform}(-0.05, -0.01)$ and $\textsc{Uniform}(0.0, 0.04)$ each with a 50\% weighting on each state $x,y,\theta$.
Figure \ref{fig:dubins1} shows the verification results in 3D and a 2D slice.
Many initial states are identified that are guaranteed to make the turn safely, or fail to meet the minimum safety threshold.
This shows the potential of the proposed method to incorporate non-standard distributions for autonomous system.

\begin{figure}
    \centering
    \begin{subfigure}{\figwidth}
        \includegraphics[trim=80px 200px 100px 200px, clip,width=\linewidth]{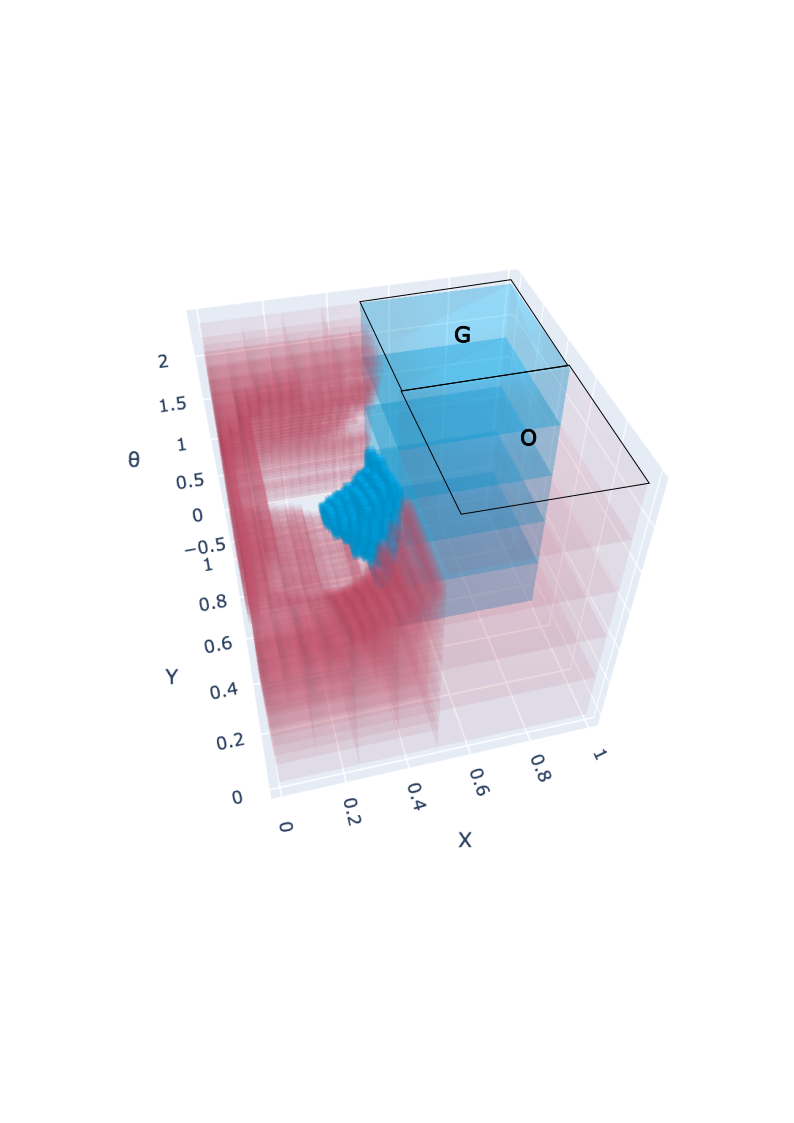}
        \subcapaboveskip
        \caption{3D Classifications}
    \end{subfigure}
    \begin{subfigure}{\figwidth}
        \includegraphics[trim=300px 0 300px 0, clip, width=\linewidth]{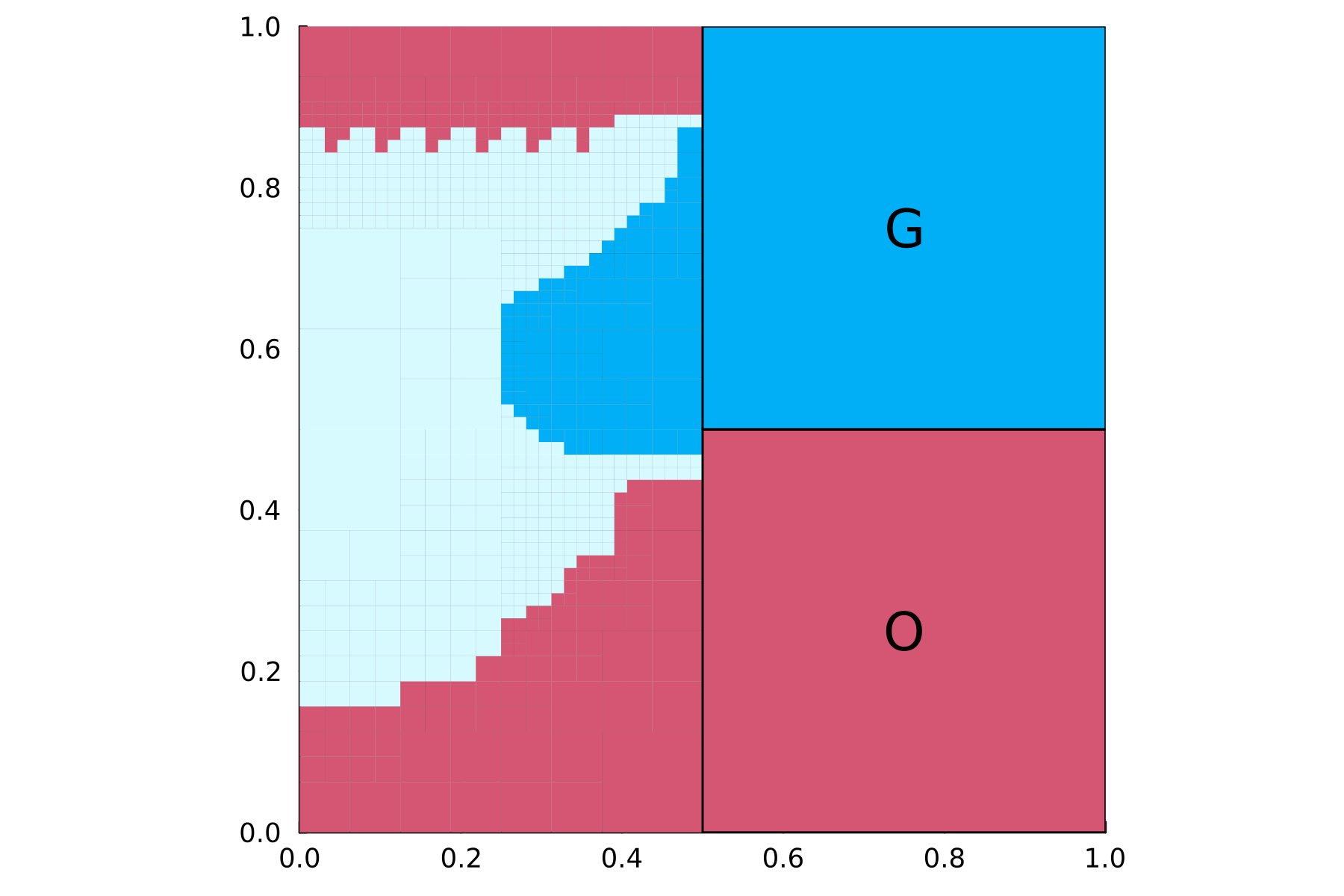}
        \subcapaboveskip
        \caption{Slice at $\theta=0.91$ (\rotatebox[origin=c]{52}{$\rightarrow$})}
    \end{subfigure}
    \caption{Verification of the constant-turn Dubin's car system.}\label{fig:dubins1}
\end{figure}

\section{CONCLUSIONS}
We present a method to construct IMC abstractions of nonlinear stochastic systems based on partitioning the noise domain, and a novel refinement-free approach to improve the verification results.
This procedure admits a wider class of systems, including those with non-affine and non-standard noise distributions, and data-driven systems.
Multiple examples demonstrate the effectiveness of the method, including nonlinear systems, multiplicative noise and a data-driven system.
Future work includes incorporating measurement models, generalizing the optimal partitioning, and improving the efficacy of the clustering procedure.

\bibliographystyle{IEEEtran}
\bibliography{root.bib}

\end{document}